\documentclass[a4paper]{article}
\usepackage{amsmath}
\usepackage{amsthm}
\usepackage{amsfonts}
\usepackage{amssymb}
\usepackage{bussproofs}
\usepackage{mathtools}
\usepackage{verbatim}
\usepackage{dsfont}
\usepackage{mathabx}
\usepackage[all, 2cell]{xy}
\usepackage[all]{xy}
\usepackage{wasysym}
\usepackage{rotating}
\usepackage{geometry}
\usepackage{trfsigns}
\usepackage{cmll}
\usepackage{authblk}
\usepackage{hyperref}
\usepackage{cleveref}
\usepackage{lipsum}

\theoremstyle{defin}
\newtheorem{defin}{Definition}

\theoremstyle{theorem}
\newtheorem{theorem}{Theorem}

\theoremstyle{prop}
\newtheorem{prop}{Proposition}

\theoremstyle{lemma}
\newtheorem{lemma}{Lemma}

\theoremstyle{ex}

\theoremstyle{col}

\date{}
\author{Daniel Rogozin}
\affil{Lomonosov Moscow State University}
\title{Quantale semantics of Lambek calculus with subexponential modalities}
\begin{document}

\maketitle

\begin{abstract}
  In this paper, we consider the polymodal version of Lambek calculus with subexponential modalities initially introduced by Kanovich,
  Kuznetsov, Nigam, and Scedrov \cite{KanovichEtc1} and its quantale semantics. In our approach, subexponential modalities have an interpretation in terms of quantic conuclei. We show that this extension of Lambek calculus is complete w.r.t quantales with quantic conuclei. Also, we prove a representation theorem for quantales with quantic conuclei and show that Lambek calculus with subexponentials is relationally complete. Finally, we extend this representation theorem to the category of quantales with quantic conuclei. Some of these results were presented here \cite{Myself}.
\end{abstract}

\section{Introduction}

Categorial grammars were initially introduced by Ajdukiewicz and Bar-Hillel \cite{Ajd} \cite{Bar}.
In the 1950s, Lambek proposed the way of proof-theoretical consideration of such grammars \cite{Lambek}. In this approach, language is logic and parsing is derivation via inference rules. Let us consider a quite simple example. Suppose one needs to parse this quote from the poem by Oscar Wilde called "Impression du Matin" \cite{Wilde}:

\begin{center}
The Thames nocture of blue and gold \\
Changed to Harmony in grey
\end{center}

The first is to assign the corresponding syntactic categories to the words as follows:

\vspace{\baselineskip}

$ $ \:\:\:\: The \:\:\:\: Thames nocturne \:\:\:\:\:\: of \:\:\:\:\:\:\:\:\:\: blue \:\:\:\:\:\:\: and \:\:\:\:\:\:\:\: gold

$(np / n) / np$ \:\:\:\:\: $np$ \:\:\:\:\:\:\:\:\: $n$ \:\:\:\:\:\: $np \backslash (np / ad)$ \: $ad$ \:\: $ad / (ad \backslash ad)$ \: $ad$

\vspace{\baselineskip}

$ $ \: Changed \:\:\:\: to \:\:\:\: Harmony \:\:\:\:\:\:\:\:\:\: in \:\:\:\:\:\:\:\:\:\:\: grey

$ $ \: $np \backslash (s / p)$ \:\: $p / np$ \:\:\:\:\:\: $np$ \:\:\:\:\:\:\:\: $np \backslash (np / ad)$ \:\: $ad$ \:\:\:\:\:\: $\rightarrow s$

\vspace{\baselineskip}

Here the basic reduction rules are:

\begin{itemize}
  \item $A, A \backslash B \to B$
  \item $B / A, A \to B$
\end{itemize}

Here we assign the special syntactic categories to the words of this sentence. $np$ denotes "noun phrase", $n$ -- noun, $ad$ -- adjective, $p$ -- phrase, $s$ -- sentence. This sequent denotes that this Oscar Wilde's quote is a well-formed sentence.
The verb "changed" has type "$np$ under ($s$ over $p$)". In other words, one needs to apply some noun phrase ("The Thames
nocturne of blue and gold") from the left, apply some phrase from the right ("Changed to Harmony in grey") and obtain
sentence after that. The other syntactic categories might be considered similarly.

The general case of such derivations in categorial grammars is axiomatised via Lambek calculus \cite{Lambek}, non-commutative
linear logic:

\begin{defin} Lambek calculus (the system $L$)
  $ $

  \begin{prooftree}
  \AxiomC{$ $}
  \RightLabel{\scriptsize{ax}}
  \UnaryInfC{$A \rightarrow A$}
  \end{prooftree}

  \begin{minipage}{0.5\textwidth}
  \begin{flushleft}
        \begin{prooftree}
      \AxiomC{$\Gamma \rightarrow A$}
      \AxiomC{$\Delta, B, \Theta \rightarrow C$}
      \RightLabel{$\backslash \rightarrow$}
      \BinaryInfC{$\Delta, \Gamma, A \backslash B, \Theta \rightarrow C$}
    \end{prooftree}

    \begin{prooftree}
      \AxiomC{$\Gamma \rightarrow A$}
      \AxiomC{$\Delta, B, \Theta \rightarrow C$}
      \RightLabel{$/ \rightarrow$}
      \BinaryInfC{$\Delta, B / A, \Gamma, \Theta \rightarrow C$}
    \end{prooftree}

    \begin{prooftree}
      \AxiomC{$\Gamma, A, B, \Delta \rightarrow C$}
      \RightLabel{$\bullet \rightarrow$}
      \UnaryInfC{$\Gamma, A \bullet B, \Delta \rightarrow C$}
    \end{prooftree}
  \end{flushleft}
  \end{minipage}\hfill
  \begin{minipage}{0.5\textwidth}
  \begin{flushright}
    \begin{prooftree}
      \AxiomC{$A, \Pi \rightarrow B$}
      \RightLabel{$\rightarrow \backslash, \Pi \text{ is non-empty}$}
      \UnaryInfC{$\Pi \rightarrow A \backslash B$}
    \end{prooftree}

    \begin{prooftree}
      \AxiomC{$\Pi, A \rightarrow B$}
      \RightLabel{$\rightarrow /, \Pi \text{ is non-empty}$}
      \UnaryInfC{$\Pi \rightarrow B / A$}
    \end{prooftree}

    \begin{prooftree}
      \AxiomC{$\Gamma \rightarrow A$}
      \AxiomC{$\Delta \rightarrow B$}
      \RightLabel{$\rightarrow \bullet$}
      \BinaryInfC{$\Gamma, \Delta \rightarrow A \bullet B$}
    \end{prooftree}
  \end{flushright}
  \end{minipage}

  \begin{prooftree}
    \AxiomC{$\Gamma \rightarrow A$}
    \AxiomC{$\Pi, A, \Delta \rightarrow B$}
    \RightLabel{${\bf cut}$}
    \BinaryInfC{$\Gamma, \Pi, \Delta \rightarrow B$}
  \end{prooftree}
\end{defin}

Note that, non-commutativity has quite strong strong linguistical connotation. We cannot assume that a sentence remains well-formed regardless of the word order.

Here, non-commutativity yields two implications: the left one (or, left division) and the right one (or, right division). Moreover, the right introduction rules for these divisions have the special restriction that claims non-emptiness
of the sequence $\Pi$ in the premise. This restriction denotes that we have no empty word that makes no sense
linguistically. On the other hand, we can avoid this restriction. This calculus is called $L^{*}$ \cite{LambekStar}. From a linguistical point of view, lack of this restriction denotes that the language has an empty word, i.e., a neutral element w.r.t product.

Note that, the empty word might be introduced explicilty via the constant ${\bf 1}$. The logic $L_{\bf 1}$ is a conservative extension of $L^{*}$ \cite{Unit} with the following inference rules for this constant:

\vspace{\baselineskip}

\begin{minipage}{0.5\textwidth}
\begin{flushleft}

  \begin{prooftree}
    \AxiomC{$\Gamma, \Delta \rightarrow A$}
    \RightLabel{${\bf 1} \rightarrow$}
    \UnaryInfC{$\Gamma, {\bf 1}, \Delta \rightarrow A$}
  \end{prooftree}

  \end{flushleft}
\end{minipage}\hfill
\begin{minipage}{0.5\textwidth}
\begin{flushright}

  \begin{prooftree}
    \AxiomC{$ $}
    \RightLabel{$\rightarrow {\bf 1}$}
    \UnaryInfC{$\rightarrow {\bf 1}$}
  \end{prooftree}

\end{flushright}
\end{minipage}

\vspace{\baselineskip}

Let us discuss models of basic Lambek calculus and its completeness theorems. The calculus with non-emptiness restriction is a
logic of residual semigroups. The logic of residual monoid is $L^{*}$. Moreover, these calculi are complete with respect
to $L$-models \cite{Bus} \cite{Pentus}, residual semigroups and monoids on subsets of free semigroup and free monoid correspondingly.

\vspace{\baselineskip}

One may extend Lambek calculus with so-called additive conjunction and disjunction. The inference rules for these connections
are quite similar to inference rules for conjunction and disjunction in intuitionistic sequent calculus. In $L$-models,
additive connections are intersection and union of languages:

\vspace{\baselineskip}

\begin{minipage}{0.5\textwidth}
\begin{flushleft}
    \begin{prooftree}
      \AxiomC{$\Gamma, A_i, \Delta \rightarrow B$}
      \RightLabel{$\land \rightarrow, i = 1,2$}
      \UnaryInfC{$\Gamma, A_1 \land A_2, \Delta \rightarrow B$}
    \end{prooftree}

    \begin{prooftree}
      \AxiomC{$\Gamma, A, \Delta \rightarrow C$}
      \AxiomC{$\Gamma, B, \Delta \rightarrow C$}
      \RightLabel{$\vee \rightarrow$}
      \BinaryInfC{$\Gamma, A \vee B, \Delta \rightarrow C$}
    \end{prooftree}
  \end{flushleft}
\end{minipage}\hfill
\begin{minipage}{0.5\textwidth}
\begin{flushright}
  \begin{prooftree}
    \AxiomC{$\Gamma \rightarrow A$}
    \AxiomC{$\Gamma \rightarrow B$}
    \RightLabel{$\rightarrow \land$}
    \BinaryInfC{$\Gamma \rightarrow A \land B$}
  \end{prooftree}

  \begin{prooftree}
    \AxiomC{$\Gamma \rightarrow A_i$}
    \RightLabel{$\rightarrow \vee, i = 1,2$}
    \UnaryInfC{$\Gamma \rightarrow A_1 \vee A_2$}
  \end{prooftree}
\end{flushright}
\end{minipage}

\vspace{\baselineskip}

Note that Lambek calculus with additives is incomplete w.r.t $L$-models. Here, there are four distributivity inclusions. They
are all true in each $L$-model, just because these properties hold for language intersection and union as for usual Boolean
operations on sets. On the other hand, only two of them are provable:

  \begin{enumerate}
    \item $\vdash A \lor (B \land C) \rightarrow (A \lor B) \land (A \lor C)$
    \item $\not\vdash (A \lor B) \land (A \lor C) \rightarrow A \lor (B \land C)$
    \item $\not\vdash A \land (B \lor C) \rightarrow (A \land B) \lor (A \land C)$
    \item $\vdash (A \land B) \lor (A \land C) \rightarrow A \land (B \lor C)$
  \end{enumerate}

One may read about the other approaches to the interpretation of Lambek calculus with additives (sometimes it's called \emph{full Lambek calculus}) here \cite{Wurm} and here \cite{Bus1}.

In this paper, we introduce quantales in order to explain the purely algebraic semantics of Lambek calculus with additives and its modal extensions. A quantale is a generalisation of locales and some well-known structures from functional analysis,
such as $C^{*}$ algebras. The initial idea to consider quantales within linear logic belongs to David Yetter \cite{Yetter}.

Brown and Gurr proved that Lambek calculus with additives is sound and strongly complete w.r.t quantales \cite{BrownAndGurr2}.
Moreover, there is a representation theorem which claims that any quantale is isomorphic to relational quantale on its underlying set \cite{BrownAndGurr1}. We discuss what relational quantale is later.

\section{Subexponential modalities}

One may extend full Lambek calculus via so-called (sub)exponential modalities and this extension might be motivated linguistically \cite{Morrill} \cite{KanovichEtc}. Let us consider the following phrase in order to explain subexponential modalities use.

\begin{center}
  The young lady whom Childe Harold met before his pilgrimage
\end{center}

There is "the young lady" in the middle of this phrase. We are incapable of processing such cases in the core Lambek
calculus. That is, there is no ability to extract from the middle. "Met" is a phrasal verb. Whom did Childe Halord meet? He
met the young lady.

For this purpose, we introduce the exponential modality with exchange rule. The left one rule is a
dereliction, which is similar to the left $\Box$ introduction in modal logic $T$.

\vspace{\baselineskip}

\begin{minipage}{0.5\textwidth}
\begin{flushleft}
\begin{prooftree}
  \AxiomC{$\Gamma, A, \Delta \rightarrow B$}
  \RightLabel{$(! \to)$}
  \UnaryInfC{$\Gamma, ! A, \Delta \rightarrow B$}
\end{prooftree}
\end{flushleft}
\end{minipage}\hfill
\begin{minipage}{0.5\textwidth}
\begin{flushright}
\begin{prooftree}
  \AxiomC{$\Gamma, ! A, \Delta, \Theta \rightarrow C$}
  \RightLabel{${\bf ex}$}
  \UnaryInfC{$\Gamma, \Delta, ! A, \Theta \rightarrow C$}
\end{prooftree}
\end{flushright}
\end{minipage}

\begin{prooftree}
  \AxiomC{$\Gamma, \Delta, ! A, \Theta \rightarrow C$}
  \RightLabel{${\bf ex}$}
  \UnaryInfC{$\Gamma, ! A, \Delta, \Theta \rightarrow C$}
\end{prooftree}

\vspace{\baselineskip}

The another one phrase is related to so-called parasitic extraction:

\begin{center}
The letter that Young Werther sent to Charlotte without reading
\end{center}

In addition to medial extraction, we used "the letter" twice in this phrase. Thus, one needs to multiply our linguistical resources in a restricted way.

\begin{center}
The letter that$_i$ Young Werther sent $e_i$ to Charlotte without reading $e_i$
\end{center}

Subexponential modality with non-local contraction allows one to do such operations:

\vspace{\baselineskip}

\begin{minipage}{0.5\textwidth}
\begin{flushleft}
\begin{prooftree}
  \AxiomC{$\Gamma, ! A, \Delta, ! A, \Theta \rightarrow B$}
  \RightLabel{${\bf contr}$}
  \UnaryInfC{$\Gamma, \Delta, ! A, \Theta \rightarrow B$}
\end{prooftree}
\end{flushleft}
\end{minipage}\hfill
\begin{minipage}{0.5\textwidth}
\begin{flushright}
\begin{prooftree}
  \AxiomC{$\Gamma, ! A, \Delta, ! A, \Theta \rightarrow B$}
  \RightLabel{${\bf contr}$}
  \UnaryInfC{$\Gamma, ! A, \Delta, \Theta \rightarrow B$}
\end{prooftree}
\end{flushright}
\end{minipage}

\vspace{\baselineskip}

Note that, the usual form of contraction yields the cut inadmissibility in contrast to the non-local contraction that generalises the contraction rule \cite{KanovichEtc2}. The similar version of contraction was also considered by de Paiva and Eades III \cite{Paiva}.

\vspace{\baselineskip}

It is useful to have many modalities and distinguish them in accordance with their abilities. That is, we are going to
consider the polymodal case. Let us introduce a subexponential signature, which is a preorder with upwardly closed subsets,
where $\mathcal{W}$ denotes weakening, and so on. Note that, a subexponential signature might have any cardinality, finite or infinite. The polymodal version of Lambek calculus with subexponential was introduced by Kanovich, Kuznetsov, Nigam, and Scedrov \cite{KanovichEtc1}.
The commutative case of subexponential modalities was considered initially by Nigam and Miller \cite{NigamMiller}.

\begin{defin} A subexponential signature is an ordered quintuple:

  $\Sigma = \langle \mathcal{I}, \preceq, \mathcal{W}, \mathcal{C}, \mathcal{E} \rangle$,
where $\langle \mathcal{I}, \preceq \rangle$ is a preorder.
$\mathcal{W}, \mathcal{C}, \mathcal{E}$ are upwardly closed subsets of $I$ and $\mathcal{W} \cap \mathcal{C} \subseteq \mathcal{E}$.
\end{defin}

The last condition claims that if there are weakening and contraction, then one may also exchange as follows.

\begin{prooftree}
  \AxiomC{$\Gamma, !_s A, \Delta, \Theta \rightarrow B$}
  \RightLabel{${\bf weak}$}
  \UnaryInfC{$\Gamma, !_s A, \Delta, !_s A, \Theta \rightarrow B$}
  \RightLabel{${\bf ncontr}$}
  \UnaryInfC{$\Gamma, \Delta, !_s A, \Theta \rightarrow B$}
\end{prooftree}
Here, $s \in \mathcal{W} \cap \mathcal{C}$.

Let us introduce the polymodal inference rules for non-commutative subexponentials. One can apply substructural
rules only if there is the relevant index on the current modality. Also, the right introduction rule is a sort of the
generalised the right $\Box$-introduction rule \'a la the modal logic ${\bf K}4$. Modality $!_s$ might be introduced on the right only if its index is less than any other
subexponential index from the antecedent. That is, if we have already used stronger modality, then we may apply the weaker one.

\begin{defin} Let $\Sigma$ be a subexponential signature. Noncommutative linear logic with subexponentials $SMALC_{\Sigma}$ is Lambek calculus ${\bf L}_{{\bf 1}}$ with additive connections and the following polymodal inference rule.

    \vspace{\baselineskip}

\begin{minipage}{0.5\textwidth}
  \begin{flushleft}
    \begin{prooftree}
    \AxiomC{$\Gamma, A, \Delta \rightarrow C$}
    \RightLabel{$! \rightarrow$}
    \UnaryInfC{$\Gamma, !^{s} A, \Delta \rightarrow C$}
    \end{prooftree}

    \begin{prooftree}
    \AxiomC{$\Gamma, !^{s} A, \Delta, !^{s} A, \Theta \rightarrow B$}
    \RightLabel{${\bf ncontr}_1, s \in C$}
    \UnaryInfC{$\Gamma, !^{s} A, \Delta, \Theta \rightarrow B$}
    \end{prooftree}

    \begin{prooftree}
    \AxiomC{$\Gamma, \Delta, !^{s} A, \Theta \rightarrow B$}
    \RightLabel{${\bf ex}_1, s \in E$}
    \UnaryInfC{$\Gamma, !^{s} A, \Delta, \Theta \rightarrow A$}
    \end{prooftree}

  \end{flushleft}
\end{minipage}
\begin{minipage}{0.5\textwidth}
  \begin{flushright}
    \begin{prooftree}
    \AxiomC{$!^{s_1} A_1, \dots, !^{s_n} A_n \rightarrow A$}
    \RightLabel{$\rightarrow !, \forall j, s_j \succeq s$}
    \UnaryInfC{$!^{s_1} A_1, \dots, !^{s_n} A_n \rightarrow !^{s} A$}
    \end{prooftree}

    \begin{prooftree}
      \AxiomC{$\Gamma, !^{s} A, \Delta, !^{s} A, \Theta \rightarrow B$}
      \RightLabel{${\bf ncontr}_2, s \in C$}
      \UnaryInfC{$\Gamma, \Delta, !^{s} A, \Theta \rightarrow B$}
    \end{prooftree}

    \begin{prooftree}
      \AxiomC{$\Gamma, !^{s} A, \Delta, \Theta \rightarrow B$}
      \RightLabel{${\bf ex}_2, s \in E$}
      \UnaryInfC{$\Gamma, \Delta, !^{s} A, \Theta \rightarrow A$}
    \end{prooftree}
  \end{flushright}
\end{minipage}

\begin{prooftree}
\AxiomC{$\Gamma, \Delta \rightarrow B$}
\RightLabel{${\bf weak}_!, s \in C$}
\UnaryInfC{$\Gamma, !^{s} A, \Delta \rightarrow B$}
\end{prooftree}

\end{defin}

Let us consider the current proof-theoretical and algorithmic results on Lambek calculus with additives and subexponentials.
First of all, the cut rule is admissible. Generally, this calculus is undecidable, but the fragment without non-local
contraction belongs to PSPACE. These results were obtained by Kanovich, Kuznetsov, Nigam, and Scedrov \cite{KanovichEtc1}:

\begin{theorem}
$ $

  \begin{enumerate}
    \item Cut-rule is admissable
    \item $\text{SMALC}_{\Sigma}$ is undecidable, if $C \neq \emptyset$
    \item If $C$ is empty, then the decidability problem of $\text{SMALC}_{\Sigma}$ belongs to PSPACE.
  \end{enumerate}
\end{theorem}

\section{Quantale background}

Now we introduce a reader to quantales quite briefly. One may take a look at these books \cite{Quantale18}
\cite{Rosenthal} in order to be familiar with quantales and related concepts closely.

\begin{defin} Quantale
$ $

  A quantale is a triple $\mathcal{Q} = \langle A, \bigvee, \cdot \rangle$, where $\langle A, \bigvee \rangle$
is a complete join semilattice and $\langle A, \cdot \rangle$ is a semigroup such that for each indexing set $J$:

\begin{enumerate}
  \item $a \cdot \bigvee \limits_{j \in J} b_i = \bigvee \limits_{j \in J} (a \cdot b_j)$;
  \item $\bigvee \limits_{j \in J} a_j \cdot b = \bigvee \limits_{j \in J} (a_j \cdot b)$
\end{enumerate}

A quanlate is called unital, if $\langle A, \cdot \rangle$ is a monoid.
\end{defin}

Note that any quantale is a complete lattice, so far as any join semilattice is a complete lattice \cite{Johnstone}.

There are several examples of quantales:

\begin{itemize}
\item Let $S$ be a semigroup (monoid), then $\langle \mathcal{P}(S), \cdot, \subseteq \rangle$
is a free (unital) quantale.
\item Let $\mathcal{R}$ be a ring and $Sub(\mathcal{R})$ be a set of additive subgroups of $R$.
We define $A \cdot B$ as an additive subgroup generated by finite sums of products $ab$ and order is defined by inclusion.
\item Any locale is a quantale with $\cdot = \wedge$.
\end{itemize}

It is easy to see, that any (unital) quantale is a residual (monoid) semigroup. We define divisions as follows:

\begin{enumerate}
\item $a \backslash b = \bigvee \{ c \: | \: a \cdot c \leq b \}$
\item $b / a = \bigvee \{ c \: | \: c \cdot a \leq b \}$
\end{enumerate}

Residuality for these divisions holds straightforwardly:

\begin{center}
  $b \leq a \backslash c \Leftrightarrow a \cdot b \leq c \Leftrightarrow a \leq b / c$
\end{center}

A quantale homomorphism, subquantales, centre are defined quite naturally:

\begin{defin}

  Let $\mathcal{Q}_1$, $\mathcal{Q}_2$ be quantales. A quantale homomorphism is a map $f : \mathcal{Q}_1 \to \mathcal{Q}_2$, such that:

  \begin{enumerate}
    \item for all $a,b \in \mathcal{Q}_1$, $f(a \cdot b) = f(a) \cdot f(b)$;
    \item for all indexing set $I$, $f(\bigvee \limits_{i \in I} a_i) = \bigvee \limits_{i \in I} f(a_i)$.
  \end{enumerate}

  If $\mathcal{Q}_1$, $\mathcal{Q}_2$ are unital quantales, then a unital homomorphism is a quantale homomorphism such that $f(\varepsilon) = \varepsilon$.
\end{defin}

\begin{defin}
$ $

  Let $\mathcal{Q} = \langle A, \bigvee, \cdot \rangle$ be a quantale. $\mathcal{S} \subseteq \mathcal{Q}$ is said to be a subquantale, if $\mathcal{S}$ is closed under multiplication and joins.
\end{defin}

\begin{defin}
$ $

  Let $\mathcal{Q} = \langle A, \bigvee, \cdot \rangle$ be a quantale.
  The centre of a quantale is the subquantale $\mathcal{Z}(\mathcal{Q}) = \{ a \in A \: | \: \forall b \in A, a \cdot b = b \cdot a \}$
\end{defin}

There occurs the following simple statement:

\begin{prop} \label{prop1}
  $ $

Let $\mathcal{Q}_1$, $\mathcal{Q}_2$ be quantales and $\mathcal{S} \subseteq \mathcal{Q}_1$ is a subquantale of
$\mathcal{Q}_1$.

Then, if $f : \mathcal{Q}_1 \to \mathcal{Q}_2$ is a quantale homomorphism, then $f(\mathcal{S}) \subseteq \mathcal{Q}_2$ is a subquantale of $\mathcal{Q}_2$.

That is, a homomorphic image of subquatale is a subquantale.
\end{prop}

\begin{proof}
$ $

  Obviously.
\end{proof}

Let us define the special kinds of elements.

\begin{defin}
  Let $\mathcal{Q}$ be a quantale and let $a \in \mathcal{Q}$:

  \begin{enumerate}
    \item $a$ is central iff $z \in \mathcal{Z}(\mathcal{Q})$
    \item $a$ is strongly square increasing iff for all $b \in \mathcal{Q}$, $a \cdot b \leq a \cdot b \cdot a$ and $b \cdot a \leq a \cdot b \cdot a$
  \end{enumerate}
\end{defin}

The strongly square increasing property is introduced by us in order to have an algebraic counterpart of non-local contraction.
The usual form of contraction corresponds to the property  that sometimes called \emph{square increasing} $a \leq a \cdot a$.
It is clear that strongly square increasing property yields the usual form of that semi-idempotence.

Let us formulate the basic properties of those elements.

\begin{lemma} \label{ssiLemma}

  Let $\mathcal{Q}$ be a unital quantale and $a \in \mathcal{Q}$ be a strongly square increasing such that $a \leq \varepsilon$. Then $a \in \mathcal{Z}(\mathcal{Q})$, that is, $a$ is central.
\end{lemma}

\begin{proof}
$ $

  $\begin{array}{lll}
  & b \cdot a \leq a \cdot b \cdot a \leq a \cdot b \cdot \varepsilon \leq a \cdot b \leq a \cdot b \cdot a \leq \varepsilon \cdot b \cdot a \leq b \cdot a&
  \end{array}$
\end{proof}

\begin{lemma} \label{ssiSub}
  Let $\mathcal{Q}$ be a quantale, the set $\mathcal{A} = \{ a \in \mathcal{Q} \: | \: \text{$a$ is strongly square increasing}\}$ is a subquantale of $\mathcal{Q}$
\end{lemma}

\begin{proof}
  One needs to prove $\mathcal{A}$ is closed under product and sups.

  \begin{enumerate}
    \item Let $J$ be a indexing set and for all $j \in J$, $a_i \in A$, and $b \in \mathcal{Q}$, then:

    $\begin{array}{lll}
    &\bigvee \limits_{j \in J} a_j \cdot b =
    \bigvee \limits_{j \in J} (a_j \cdot b) \leq
    \bigvee \limits_{j \in J} (a_j \cdot b \cdot a_j) =
    \bigvee \limits_{j \in J} a_j \cdot \bigvee \limits_{j \in J} (b \cdot a_j) =
    \bigvee \limits_{j \in J} a_j \cdot b \cdot \bigvee \limits_{j \in J} a_j & \\
    \end{array}$

    The second inequation $b \cdot \bigvee \limits_{j \in J} a_j \leq \bigvee \limits_{j \in J} a_j \cdot b \cdot \bigvee \limits_{j \in J} a_j$ might be proved similarly.

    \item Let $a_1, a_2 \in \mathcal{A}$. Let us show that $a_1 \cdot a_2 \in A$.

    First of all, $a_1 \cdot a_2 \cdot b = a_1 \cdot (a_2 \cdot b) \leq a_1 \cdot (a_2 \cdot b) \cdot a_1$.
    By monotonicity, $a_1 \cdot a_2 \cdot b \cdot a_2 \leq a_1 \cdot a_2 \cdot b \cdot a_1 \cdot a_2$.

    On the other hand, $a_1 \cdot a_2 \cdot b \leq a_1 \cdot a_2 \cdot b \cdot a_2$.
    Then, by transitivity, $a_1 \cdot a_2 \cdot b \leq a_1 \cdot a_2 \cdot b \cdot a_1 \cdot a_2$.

    The second inequation $b \cdot (a_1 \cdot a_2) \leq (a_1 \cdot a_2) \cdot b \cdot (a_1 \cdot a_2)$ might be proved similarly.
  \end{enumerate}
\end{proof}

\begin{lemma} \label{Unital}

  Let $\mathcal{Q}$ be a unital quantale, then the set $\mathcal{A} = \{ a \in \mathcal{Q} \: | \: a \leq \varepsilon \}$
\end{lemma}

\begin{proof} Straightforwardly.
\end{proof}

We introduce the notion of quantic conucleus, an interior (or coclosure) operator on quantale with additional axiom for a product.

\begin{defin}
$ $

  A quantic conucleus on quantale $\mathcal{Q}$ is a map $I : \mathcal{Q} \to \mathcal{Q}$ such that

\begin{enumerate}
  \item $I a \leq a$;
  \item $I a = I^2 a$;
  \item $a \leq b \Rightarrow I a \leq I b$;
  \item $I a \cdot I b = I (I a \cdot I b)$.
\end{enumerate}

For unital quantale, we require that $I \varepsilon = \varepsilon$.
\end{defin}

Here and below, $I^2 a$ denotes $I (I a)$.

Note that, one may replace the last condition to $I a \cdot I b \leq I (a \cdot b)$. Thus, a quantic conuleus is the special case of lax monoidal comonad from a category-theoretic point of view.

An element $a$ is \emph{open} if and only if $I a = a$. As usual \cite{Rasowa}, if $a$ is open, then:

\begin{center}
  $a \leq b$ if and only if $a \leq I b$
\end{center}

One may extend homomorphism between quantales to homomorphism between quantales with quantic conuclei via the additional
condition $f(I_1 a) = I_2 (f a)$

Also, we define a pointwise order on quantic conuclei, that is, $I_1 \leq I_2 \Leftrightarrow \forall a \: I_1 a \leq I_2 a$.

\begin{lemma} \label{ManyCo}
  $ $
  Let $\mathcal{Q}$ be a quantale. Let $I$, $I_1$, and $I_2$ be quantic conuclei. Then $I_1 a_1 \cdot I_2 a_2 \leq I (I_1 a_1 \cdot I_2 a_2)$, where $I_i \leq I, i = 1,2$.
\end{lemma}

\begin{proof}
  $I_1 a_1 \cdot I_2 a_2 \leq I_1 (I_1 a_1) \cdot I_2 (I_2 a_2) \leq I (I_1 a_1) \cdot I (I_2 a_2) \leq I (I_1 a_1 \cdot I_2 a_2)$
\end{proof}

The following lemma allows one to introduce some quantic conulei via subquantales and compare them.

\begin{lemma} \label{lemmaConuc}
$ $

  \begin{enumerate}
    \item Let $S \subseteq \mathcal{Q}$ be a subquantale, then the operation $I : \mathcal{Q} \to \mathcal{Q}$, such that $I a = \bigvee \{ q \in S \: | \: q \leq a \}$, is a quantic conucleus.
    \item Let $S_1$ and $S_2$ are subquantales such that $\mathcal{S}_1 \subseteq \mathcal{S}_2 \subseteq \mathcal{Q}$. Then $I_{\mathcal{S}_1} (a) \leq I_{\mathcal{S}_1} (a)$, where $I_{\mathcal{S}_i}, i = 1,2$ are quantic conucleis defined according to the previous part.
  \end{enumerate}
\end{lemma}

\begin{proof}
  For the first part, see \cite{Rosenthal}. The second one is proved immediately.
\end{proof}

It is well known that if $I$ is a quantic conucleus on quantale $\mathcal{Q}$,
then the set $\mathcal{Q}_I = \{ a \in \mathcal{Q} \: | \: I a = a \}$ is a subquantale of $\mathcal{Q}$ \cite{Rosenthal}. Let us formulate and prove the following statement.

\begin{lemma} \label{ConucUnfold}
$ $

  Let $I$ be a quantic conucleus, then for all $a \in \mathcal{Q}$, $I a = \bigvee \{ q \in Q_{I} \: | \: q \leq a \}$, where $Q_{I} = \{ a \in \mathcal{Q} \: | \: I a = a \}$.
\end{lemma}

\begin{proof}
$ $

So far as $Q_{I} = \{ a \in \mathcal{Q} \: | \: I a = a \}$ is a subquantale, then $I' a = \bigvee \{ q \in Q_{I} \: | \: q \leq a \}$ is a quantic conucleus by~\cref{lemmaConuc}. Let us show that these operations are equivalent.

From the one hand, $\bigvee \{ q \in Q_{I} \: | \: q \leq a \} = \bigvee \{ q \in Q_{I} \: | \: q \leq I a \} \leq I a$, so far as $Q_{I}$ is a subquantale of open elements.

On the other hand, $I a \leq a$, then $I a \in \{ q \in Q_{I} \: | \: q \leq a \}$, so far as $I a \in Q_{I}$ (by idempotence). Thus, $I a \leq \bigvee \{ q \in Q_{I} \: | \: q \leq a \}$.

\end{proof}

Let us define the special cases of quantic conuclei that would be quite useful for our purpose:

\begin{defin}
$ $
\begin{enumerate}
  \item A quantic conucleus $I$ is called central, if for all
  $a, b \in \mathcal{Q}, I a \cdot b = b \cdot I a$.

  \item A quantic conucleus is called strongly square-increasing,
  if for all $a, b \in \mathcal{Q}, I a \cdot b \leq I a \cdot b \cdot I a$ and
  $b \cdot I a \leq I a \cdot b \cdot I a$.

  \item A quantic conucleus is called unital, if for each $a \in Q, I a \leq \varepsilon$.
\end{enumerate}
\end{defin}

There are the following properties of those modalities:

\begin{lemma} \label{Upward}
$ $

Let $I_1$ be a quantic conucleus and $I_2$ an unital (central, strongly weak increasing) quantic conucleus such that $I_1 \leq I_2$,
then $I_1$ is an unital (central, strongly weak increasing).

\end{lemma}

\begin{proof}
  $ $

\begin{enumerate}
  \item Let $a \in \mathcal{Q}$ (here $\mathcal{Q}$ is supposed to be a unital).
  Then $I_2 a \leq \varepsilon$, but $I_1 \leq I_2$, so $I_1 a \leq I_2 a$. By transitivity, $I_1 a \leq \varepsilon$.
  \item Let $a \in \mathcal{Q}$ and $I_2$ is a central quantic conucleus. Let $I_1$ be a quantic conucleus and $I_1 \leq I_2$.

  Thus, $I_2 a \in \mathcal{Z}(\mathcal{Q})$ and, obviously, $Q_{I_2} \subseteq \mathcal{Z}(\mathcal{Q})$.

  Let $b \in Q_{I_1}$, then $b = I_1 b \leq I_2 b \leq b$. Thus, $b = I_2 b$ and $b \in Q_{I_2}$. So, $Q_{I_1} \subseteq \mathcal{Z}(\mathcal{Q})$.

  Let us check that $I_1$ commutes with each other element $d \in \mathcal{Q}$

  $\begin{array}{lll}
  & I_1 a \cdot d = & \\
  & \:\:\:\:\:\:\:\: \text{By~\cref{ConucUnfold}}& \\
  & \bigvee \{ q \in Q_{I_1} \: | \: q \leq a \} \cdot d = & \\
  & \:\:\:\:\:\:\:\: \text{By quantale axiom}& \\
  & \bigvee \{ q \cdot d \: | \: q \in Q_{I_1} \: \& \: q \leq a \} = & \\
  & \:\:\:\:\:\:\:\: \text{As we proved, $q \in \mathcal{Z}(\mathcal{Q})$} & \\
  & \bigvee \{ d \cdot q \: | \: q \in Q_{I_1} \: \& \: q \leq a \} = & \\
  & \:\:\:\:\:\:\:\: \text{By quantale axiom}& \\
  & d \cdot \bigvee \{ q \in Q_{I_1} \: | \: q \leq a \} = & \\
  & \:\:\:\:\:\:\:\: \text{By~\cref{ConucUnfold}}& \\
  & d \cdot I_1 a &
  \end{array}$

  \item Let $I_2$ be a strongly square increasing. Then, for all $a \in \mathcal{Q}$ and $b \in Q_{I_2}$,
  $a \cdot b \leq b \cdot a \cdot b$ and $b \cdot a \leq b \cdot a \cdot b$.

  Let $c \in Q_{I_1}$. Similarly to the previous part, $c = I_1 c \leq I_2 c \leq c$, then $c \in Q_{I_2}$ and
  $I_1$ is a strongly square increasing:

  $\begin{array}{lll}
  &I_1 a \cdot b = & \\
  & \:\:\:\:\:\:\:\: \text{By~\cref{ConucUnfold}}& \\
  & \bigvee \{ q \in Q_{I_1} \: | \: q \leq a \} \cdot b = & \\
  & \:\:\:\:\:\:\:\: \text{By quantale axiom}& \\
  & \bigvee \{ q \cdot b \: | \: q \in Q_{I_1}, q \leq a \} = & \\
  & \:\:\:\:\:\:\:\: \text{As we proved above}& \\
  & \bigvee \{ I_2 q \cdot b \: | \: q \in Q_{I_1}, q \leq a \} \leq & \\
  & \:\:\:\:\:\:\:\: \text{By condition}& \\
  & \bigvee \{ I_2 q \cdot b \cdot I_2 q \: | \: q \in Q_{I_1}, q \leq a \} = & \\
  & \:\:\:\:\:\:\:\: \text{So far as $q$ is a fixpoint of $I_2$} & \\
  & \bigvee \{ q \cdot b \cdot q \: | \: q \in Q_{I_1}, q \leq a \} = & \\
  & \:\:\:\:\:\:\:\: \text{By quantale axiom} = & \\
  & \bigvee \{ q \cdot b \: | \: q \in Q_{I_1}, q \leq a \} \cdot \bigvee \{ q \in Q_{I_1} \: | \: q \leq a \} = & \\
  & \:\:\:\:\:\:\:\: \text{By quantale axiom} = & \\
  & \bigvee \{ q \in Q_{I_1} \: | \: q \leq a \} \cdot b \cdot \bigvee \{ q \in Q_{I_1} \: | \: q \leq a \} = & \\
  & \:\:\:\:\:\:\:\: \text{By~\cref{ConucUnfold}} = & \\
  &I_1 a \cdot b \cdot I_1 a&
  \end{array}$

  The second inequation might be proved similarly.
\end{enumerate}
\end{proof}

The following proposition claims that quantic conuclei might be introduced via subquantales with properties considered above.

\begin{prop} \label{Con1}
$ $

Let $\mathcal{Q}$ be a quantale and $\mathcal{S} \subseteq \mathcal{Q}$ a subquantale, then the following operations are open modalities:

\begin{enumerate}
  \item $I a = \bigvee \{ q \in S \: | \: q \leq a, q \in \mathcal{Z}(\mathcal{Q}) \}$ is a central quantic conucleus
  \item $I a = \bigvee \{ q \in S \: | \: q \leq a, q \leq \varepsilon \}$ is an unital conucleus
  \item $I a = \bigvee \{ q \in S \: | \: q \leq a, \forall b \in Q, b \cdot s \vee s \cdot b \leq s \cdot b \cdot s\}$ is a strongly square increasing quantic conucleus
  \item Operations with combinations of conditions above.
\end{enumerate}
\end{prop}

\begin{proof}
  Immediately, so far as central, strongly square increasing, and elements that less or equal to identity form a subquantale by~\cref{ssiSub} and~\cref{Unital}.
\end{proof}

As we showed above, any strongly square increasing element that less or equal to identity belongs to the centre.
That statement might be formulated in terms of quantic conuclei as follows:

\begin{prop} \label{Con2}
If $I$ is unital and strongly square increasing, then $I$ is central.
\end{prop}
\begin{proof} Follows from~\cref{ssiLemma}.
\end{proof}

\section{Subexponential interpretation and soundness}

We define a subexponential interpretation quite sophisticatedly. Suppose we have a contravariant map from a given subexponential
signature to the set of subquantales $Sub(\mathcal{Q})$.
Here, contravariance denotes that the stronger subexponential index (in sense of preorder) maps to the weaker subquantale. In other words, the stronger modality has a smaller domain in a sense of inclusion.
The second one function maps a subquantale to its quantic conucleus according to the previous proposition. We match the index pursuant
to its subset. In other words, if $s \in W$, then the result of subexponential interpretation is an unital quantic conucleus, etc.

\begin{defin} Interpretation of subexponential signature

Let $\mathcal{Q}$ be a quantale and $\Sigma = \langle \mathcal{I}, \preceq, \mathcal{W}, \mathcal{C}, \mathcal{E} \rangle$
a subexponential signature. Let $S : \Sigma \to Sub(\mathcal{Q})$ be a contravariant map from this subexponential
signature to the set of subquantales of $\mathcal{Q}$ Thus, a subexponential interpretation is the map
$\sigma : \Sigma \to CN(\mathcal{Q})$, where $CN(\mathcal{Q})$ is the set of quantic conuclei on $\mathcal{Q}$, such that:

  $\sigma(s_i) = \begin{cases}
  I_i : \mathcal{Q} \to \mathcal{Q} \text{ s.t. } \forall a \in Q, I_i(a) = \{ q \in S_i \: | \: q \leq a\},
  \\ \:\:\:\: \text{if $s_i \notin W \cap C \cap E$} \\
  I_i : \mathcal{Q} \to \mathcal{Q} \text{ s.t. } \forall a \in Q, I_i(a) = \{ q \in S_i \: | \: q \leq a, q \leq \mathds{1}\},
  \\ \:\:\:\: \text{if $s_i \in W$} \\
  I_i : \mathcal{Q} \to \mathcal{Q} \text{ s.t. } \forall a \in Q, I_i(a) = \{ q \in S_i \: | \: q \leq a, s \in \mathcal{Z}(\mathcal{Q}) \},
  \\ \:\:\:\: \text{if $s_i \in E$} \\
  I_i : \mathcal{Q} \to \mathcal{Q} \text{ s.t. } \forall a \in Q, I_i(a) = \{ q \in S_i \: | \: q \leq a, \forall b, b \cdot q \vee q \cdot b \leq q \cdot b \cdot s \},
  \\ \:\:\:\: \text{if $s_i \in E$} \\
  \text{otherwise, if $s_i$ belongs to some intersection of subsets, then we combine the relevant conditions } \\
  \end{cases}$
\end{defin}

The resulting operators are well-defined quantic conuclei according to~\cref{Con1} and~\cref{Con2}.

First of all, let us make sure that any subexponential interpretation is sound w.r.t the set of quantic conuclei $CN(\mathcal{Q})$ of the quantale $\mathcal{Q}$.

\begin{prop} \label{ImageContr}
  $s \in A = \mathcal{W}, \mathcal{C}, \mathcal{E}$ iff $\sigma(s)$ is an unital (central or square increasing) modality.
\end{prop}

\begin{proof}
  By construction.
\end{proof}

\begin{lemma} \label{WellDef}
  $ $

  Let $\Sigma$ be a subexponential signature and $\sigma : \Sigma \to CN(\mathcal{Q})$ for a contravariant map
  $S : \Sigma \to Sub(\mathcal{Q})$, then:
  \begin{enumerate}
  \item For all $s_1, s_2 \in \Sigma$, if $s_1 \preceq s_2$, then $\sigma(s_2) \leq \sigma(s_1)$.
  In other words, $\sigma$ is a contravariant functor from $\Sigma$ to $Sub(\mathcal{Q})$.
  \item Let $A = \mathcal{W}, \mathcal{C}, \mathcal{E}$, then $A$ is downwardly closed
  $\sigma(A) = \{ \sigma(s) \: | \: s \in A \}$ is downwardly closed.
  \item If $s \in \mathcal{W} \cap \mathcal{C}$, then $\sigma(s)$ is a central quantic conucleus.
\end{enumerate}
\end{lemma}

\begin{proof}
  $ $

  \begin{enumerate}
  \item Let $s_1, s_2 \in \Sigma$ and $s_1 \leq s_2$. Then $S(s_2) \subseteq S(s_1)$.
  So, $\forall a \in \mathcal{Q}, \bigvee \{ s \in S(s_2) \: | \: s \leq a \} \leq \bigvee \{ s \in S(s_1) \: | \: s \leq a \}$ by~\cref{lemmaConuc}. Thus, $\sigma(s_2) \leq \sigma(s_1)$.
  \item Let $\sigma(s) \in \sigma(A)$ and $\sigma(s_1) \leq \sigma(s)$, then $\sigma(s)$ is an unital (central or square increasing) by~\cref{ImageContr}. By~\cref{Upward}, $\sigma(s_1)$ is an unital (central or square increasing) too. Thus, $s_1 \in A$ and $\sigma(s_1) \in \sigma(A)$.
  \item Let $s \in \mathcal{W} \cap \mathcal{C}$ and $a \in \mathcal{Q}$,
  then $\sigma(s)(a)$. Thus, by~\cref{Con2}, $\sigma(s)(a) \in \mathcal{Z}(\mathcal{Q})$.

\end{enumerate}
\end{proof}

An interpretation and an entailment relation are defined standardly via valuation map and subexponential interpretation.

\begin{defin} Let $\mathcal{Q}$ be an unital quantale, $f : Tp \to \mathcal{Q}$ a valuation and $\sigma : \Sigma \to CN(\mathcal{Q})$ a subexponential interpretation, then interpretation is defined inductively:

\begin{center}
$\begin{array}{lll}
& [\![p_i]\!] = f(p_i)&\\
& [\![\mathds{1}]\!] = e & \\
&[\![A \bullet B]\!] = [\![A]\!] \cdot [\![B]\!] & \\
&[\![A \backslash B]\!] = [\![A]\!] \backslash [\![B]\!] & \\
&[\![A / B]\!] = [\![A]\!] / [\![B]\!]& \\
&[\![A \& B]\!] = [\![A]\!] \wedge [\![B]\!]& \\
&[\![A \vee B]\!] = [\![A]\!] \vee [\![B]\!]& \\
&[\![!_{s_i} A]\!] = \sigma(s_i) [\![A]\!]&
\end{array}$
\end{center}
\end{defin}

\begin{defin}
  $\Gamma \models A \Leftrightarrow \forall f, \forall \sigma, [\![\Gamma]\!] \leq [\![A]\!]$
\end{defin}

\begin{theorem}
$ $

  Let $\Sigma$ be a subexponential signature, then
  $SMALC_{\Sigma} \vdash \Gamma \rightarrow A \Rightarrow [\![\Gamma]\!] \leq [\![A]\!]$.
\end{theorem}

\begin{proof}
  $ $

Let $\sigma$ be a subexponential interpretation. By~\cref{WellDef}, $\sigma$ is well-defined. Let us consider the modal cases.

\begin{enumerate}
\item Let $!_{s_1} A_1, \dots, !_{s_n} A_n \rightarrow A$ and $\forall i, s \preceq s_i$. Then $\forall a \in Q, \sigma(s_i)(a) \leq \sigma(s)(a)$ by~\cref{WellDef}.

By IH, $\sigma(s_1)[\![A_1]\!] \cdot \dots \cdot \sigma(s_n) [\![A_n]\!] \leq [\![A]\!]$. Thus,
$\sigma(s)(\sigma(s_1)[\![A_1]\!] \cdot \dots \cdot \sigma(s_n) [\![A_n]\!]) \leq \sigma(s)([\![A]\!])$. By~\cref{ManyCo},
$\sigma(s_1)[\![A_1]\!] \cdot \dots \cdot \sigma(s_n) [\![A_n]\!] \leq \sigma(s)(\sigma(s_1)[\![A_1]\!] \cdot \dots \cdot
\sigma(s_n) [\![A_n]\!])$. \\
So, $\sigma(s_1)[\![A_1]\!] \cdot \dots \cdot \sigma(s_n) [\![A_n]\!] \leq \sigma(s)([\![A]\!])$.
\item Let $s \in \mathcal{W}$ and $\Gamma, \Delta \rightarrow A$.
By IH, $[\![\Gamma]\!] \cdot [\![\Delta]\!] \leq [\![B]\!]$.
Thus, obviously, $[\![\Gamma]\!] \cdot \varepsilon \cdot [\![\Delta]\!] \leq [\![B]\!]$.
$s \in \mathcal{W}$, then $\sigma(s)$ is an unital quantic conulceus by~\cref{ImageContr}. So, for each formula $A$, $\sigma(s)([\![A]\!]) \leq \varepsilon$. Then, $[\![\Gamma]\!] \cdot \sigma(s)([\![A]\!]) \cdot [\![\Delta]\!] \leq [\![B]\!]$.
\item Let $s \in \mathcal{E}$ and $\Gamma, !_s A, \Delta, \Theta \rightarrow B$.
By IH, $[\![\Gamma]\!] \cdot \sigma(s)([\![A]\!]) \cdot [\![\Delta]\!] \cdot [\![\Theta]\!] \leq [\![B]\!]$. By~\cref{ImageContr},
$\sigma(s)$ is a central quantic conucleus,
then $[\![\Gamma]\!] \cdot [\![\Delta]\!] \cdot \sigma(s)([\![A]\!]) \cdot [\![\Theta]\!] \leq [\![B]\!]$.
\item Let $s \in \mathcal{C}$ and $\Gamma, !_s A, \Delta, !_s A, \Theta \rightarrow B$.
So, $[\![\Gamma]\!] \cdot \sigma(s)([\![A]\!]) \cdot [\![\Delta]\!] \cdot \sigma(s)([\![A]\!]) \cdot [\![\Theta]\!] \leq [\![B]\!]$.
But $\sigma(s)$ is a strongly square increasing, so $\sigma(s)([\![A]\!]) \cdot [\![\Delta]\!] \leq \sigma(s)([\![A]\!]) \cdot [\![\Delta]\!] \cdot \sigma(s)([\![A]\!])$ and
$[\![\Delta]\!] \cdot \sigma(s)([\![A]\!]) \leq \sigma(s)([\![A]\!]) \cdot [\![\Delta]\!] \cdot \sigma(s)([\![A]\!])$.
Thus, $[\![\Gamma]\!] \cdot \sigma(s)([\![A]\!]) \cdot [\![\Delta]\!] \cdot [\![\Theta]\!] \leq [\![B]\!]$ and
$[\![\Gamma]\!] \cdot [\![\Delta]\!] \cdot \sigma(s)([\![A]\!]) \cdot [\![\Theta]\!] \leq [\![B]\!]$
\end{enumerate}
\end{proof}

\section{Completeness}

Now we prove completeness a la MacNeille completetion. For Lambek calculus with additives, the completeness was proved similarly by Brown and Gurr \cite{BrownAndGurr2}. As a matter of fact, we generalise this technique for the polymodal case.

\begin{defin}
$ $

  Let $\mathcal{F} \subseteq Fm$, an ideal is a subset $\mathcal{I} \subseteq \mathcal{F}$, such that:

\begin{itemize}
  \item If $B \in \mathcal{I}$ and $A \rightarrow B$, then $A \in \mathcal{I}$
  \item If $A, B \in \mathcal{I}$, then $A \lor B \in \mathcal{I}$
\end{itemize}
\end{defin}

\begin{lemma}
$ $

  \begin{enumerate}
  \item $\bigvee S$ is an ideal
  \item $A \subseteq Fm$, then $\bigvee \{ A \} = \{ B \: | \: B \rightarrow A \}$
  \item $\bigvee \{ A \} \subseteq \bigvee \{ B \}$ iff $A \rightarrow B$
  \item Let $\mathcal{Q} = \{ \bigvee S \: | \: S \subseteq Fm \}$ and $\bigvee \mathcal{A} \cdot \bigvee \mathcal{B} =
  \bigvee \{ A \bullet B \: | \: A \in \mathcal{A}, B \in \mathcal{B} \}$
  Then $\langle \mathcal{Q}, \subseteq, \cdot, \bigvee{{\bf 1}}\rangle$ is an unital quantale.
  \end{enumerate}
\end{lemma}

\begin{proof} See \cite{BrownAndGurr2}.
\end{proof}

We extend this construction for polymodal case as follows:

\begin{lemma} \label{SynContr}
$ $

Let $\Sigma$ be a subexponential signature.
  \begin{enumerate}
  \item Let $s \in \Sigma$, then $!_s Syn = \{ \bigvee !_s S \: | \: A \subseteq Fm \}$ is a subquantale of $Syn$.
  \item Let $\Sigma$ be a subexponential signature, then a map $S : \Sigma \to Sub(Syn)$ such that $S(s) = !_s Syn$ is a contravariant map.
  \end{enumerate}
\end{lemma}

\begin{proof}
$ $

\begin{enumerate}
    \item One need to check that $!_s Syn = \{ \bigvee !_s A \: | \: S \subseteq Fm \}$ is closed under product and sups. Let us show that $!_s Syn$ is closed under product, so far as $!_s Syn$ is closed under sups by construction.

    Let $\bigvee !_s X, \bigvee !_s Y \in !_s Syn$. Then $\bigvee !_s X \cdot \bigvee !_s Y = \bigvee (!_s X \cdot !_s Y) \in !_s Syn$, so far as $!_s X \cdot !_s Y \subseteq Fm$.

    \item Let $s_1 \preceq s_2$ and $\bigvee !_{s_2} X \in \: !_{s_2} Syn$. Let us show that $\bigvee !_{s_2} X \in !_{s_1} Syn$.

    $\begin{array}{lll}
    & \bigvee !_{s_2} X = \{ !_{s_2} B \: | \: !_{s_2} B \to A, A \in X \} = & \\
    & \{ !_{s_2} B \: | \: !_{s_2} B \to !_{s_2} A, A \in X \} \subseteq & \\
    & \{ !_{s_2} B \: | \: !_{s_2} B \to !_{s_2} A, A \in X \} \subseteq \bigvee !_{s_1} X& \\
    \end{array}$
  \end{enumerate}
\end{proof}

\begin{lemma}
  $ $

Let $\Sigma$ be a subexponential signature.

\begin{enumerate}
\item Let $S$ be a contravariant map from the previous stament and $!_s \in \Sigma$ and $A \in \mathcal{F}_{\Sigma}$, then $S(s)$ is a quantic conucleus.
\item Let $A \in \mathcal{F}_{\Sigma}$, then $\Box_s \bigvee \{ A \} = \bigvee \{ !_s A \}$, for each $s \in \mathcal{I}$
\item  Let $i, j \in \Sigma$ and $i \preceq j$, then for all $A \in \mathcal{F}_{\Sigma}$, $\Box_j (\bigvee \{ A \}) \subseteq \Box_i (\bigvee \{ A \})$.
\end{enumerate}

\end{lemma}

\begin{proof}
$ $

\begin{enumerate}
\item Similarly to \cite{Yetter}.
\item Let $A \in Fm$ and $s \in \Sigma$.

Let $!_s B \in \Box_s \bigvee \{ A \}$, then $!_s B \rightarrow A$, then $!_s B \rightarrow !_s A$
by promotion. So, $!_s B \in \bigvee \{ !_s A \}$.

Let $C \in \bigvee \{ !_s A \}$, then $C \rightarrow !_s A$, so $!_s C \rightarrow !_s A$ by dereliction, but $!_s A \rightarrow A$, hence $!_s C \rightarrow A$ by cut. So, $!_s C \in \Box_s \bigvee \{ A \}$.
\item   Let $i, j \in I$ and $i \preceq j$, then forall $A \in \mathcal{F}_{\Sigma}$, $!_j A \rightarrow !_i A$ by promotion.
Then $\bigvee \{ !_j A \} \subseteq \bigvee \{ !_i A \}$, so $\Box_j (\bigvee \{ A \}) \subseteq \Box_i (\bigvee \{ A \})$.
\end{enumerate}
\end{proof}

\begin{lemma}
$ $

For all $A \in \mathcal{F}_{\Sigma}$,
  \begin{enumerate}
    \item Let $s \in W$, then $\Box_s \{ A \} \subseteq \{ {\bf 1} \}$
    \item Let $s \in E$, then $\Box_s (\bigvee \{ A \}) \cdot \bigvee \{ B \} = \bigvee \{ B\} \cdot \Box_s (\bigvee \{ A \})$
    \item Let $s \in C$, then $(\Box_s \bigvee A \cdot \bigvee B) \cup (\bigvee B \cdot \Box_s \bigvee A) \subseteq \Box_s \bigvee A \cdot \bigvee B \cdot \Box_s \bigvee A$, for all $B \subseteq Fm$.
  \end{enumerate}
\end{lemma}

\begin{proof}
$ $

\begin{enumerate}
  \item Follows from $!_s A \rightarrow {\bf 1}$, since $s \in W$.
  \item Follows from $!_s A \bullet B \leftrightarrow B \bullet !_s A$.
  \item Follows from $SMALC_{\Sigma} \vdash !_s A \bullet B \rightarrow !_s A \bullet B \bullet !_s A$ and $SMALC_{\Sigma}
  \vdash B \bullet !_s A  \rightarrow !_s A \bullet B \bullet !_s A$.
\end{enumerate}
\end{proof}

\begin{defin}
$ $

Let $\mathcal{Q}$ be a syntactic quantale defined above and $\Sigma = \langle I, \preceq,
\mathcal{W}, \mathcal{C}, \mathcal{E} \rangle$ be a subexponential signature and $S : \Sigma \to Sub(Syn)$ be a contravariant map in sense of~\cref{SynContr}.

We define a map $\Box : \Sigma \to CN(Syn)$ as follows:

$\Box(s)(\bigvee \{ A \} ) = \{ !_s B \in S(Syn) \: | \: !_s B \rightarrow A \}$.
\end{defin}

\begin{lemma} $\Box$ is a subexponential interpretation.
\end{lemma}

\begin{proof}
  Follows from lemmas above.
\end{proof}

\begin{lemma}
$ $

  Let $Q$ be a quantale constructed above, then there exist a model $\langle Q, [\![.]\!] \rangle$, such that $[\![A]\!] = \bigvee \{ A \}$, $A \in Fm$.
\end{lemma}

\begin{proof}
$ $

  We define an interpreation as follows:

\begin{enumerate}
  \item $[\![p_i]\!] = \bigvee \{ p_i \}$
  \item $[\![{\bf 1}]\!] = \bigvee \{ {\bf 1} \}$
  \item $[\![A \bullet B]\!] = \bigvee \{ A \bullet B \}$
  \item $[\![A / B]\!] = \bigvee \{ A / B \}$
  \item $[\![B \setminus A]\!] = \bigvee \{ B \setminus A \}$
  \item $[\![A \& B ]\!] = \bigvee \{ A \& B \}$
  \item $[\![A \lor B]\!] = \bigvee \{ A \lor B\}$
  \item $[\![!_s A]\!] = \Box(s) (\bigvee \{ A \}) = \{ !_s B \in S(s) \: | !_s B \rightarrow A \} = \{ !_s B \in S(s) \: | !_s B \rightarrow A \} = \bigvee \{ !_s A \}$.
\end{enumerate}
\end{proof}

\begin{theorem}
  $\Gamma \models A \Rightarrow \Gamma \rightarrow A$.
\end{theorem}

\begin{proof}
  Follows from lemmas above.
\end{proof}

\section{Representation theorem and relational completeness}

A relational quantale was initially introduced by Brown and Gurr \cite{BrownAndGurr1}.
Moreover, there is a representation in the same paper which claims that any quantale is isomorphic on its underlying set.

\begin{defin}
  $ $

  Let $A$ be a set and $R \subseteq A \times A$ a transitive relation on $A$.
  Then relational quantale on $A$ is a triple $\mathcal{Q} = \langle \mathcal{P}(R), \bigvee, \mathcal{I}
  \rangle$ with the following data:

  \begin{enumerate}
    \item $\langle \mathcal{P}(R), \bigvee, \subseteq \rangle$ is a complete semi-lattice
    \item Multiplication is defined as
      $R \circ S = \{ \langle a, c \rangle \: | \: \exists b \in A, \langle a, b \rangle \in R \text{ and } \langle b, c \rangle \in S\}$
    \item $\langle \mathcal{A}, \circ, \mathcal{I} \rangle$ is a monoid, that is, $\mathcal{I}$ is identity
    \item For each indexing set $J$, $R \circ \bigvee_{j \in J} S_j = \bigvee_{j \in J} (R \circ S_j)$ and
      $\bigvee_{j \in J} R_j \circ S = \bigvee_{j \in J}(R_j \circ S)$.
  \end{enumerate}
\end{defin}

\begin{lemma} \label{Representation lemma}
  $ $

  Let $\mathcal{Q} = \langle A, \leq, \cdot, \bigvee \rangle$ be a unital quantale and $\mathcal{S}$ is a subquantale of $\mathcal{Q}$.

  Then $\langle \mathcal{Q}, I_{\mathcal{S}} \rangle$ is isomorphic to some relational quantale of $A$
  with some quantic conucleus $\hat{I}$, where $I_{\mathcal{S}}$ is a quantic conclues on $\mathcal{S}$.
\end{lemma}

\begin{proof}
$ $

  We are piggybacking on the construction provided by Brown and Gurr \cite{BrownAndGurr1}.

  This quantale is 4-tuple $\hat{\mathcal{Q}} = \langle \mathcal{R}, \subseteq, \circ, \bigvee \rangle$ defined as follows, where
  $\mathcal{R} = \{ \hat{a} \: | \: a \in \mathcal{Q}\}$ and $\hat{a} = \{ \langle b, c \rangle \: | \: b \leq a \cdot c \}$.
  One may prove that:
  \begin{enumerate}
    \item $\widehat{a \cdot b} = \hat{a} \circ \hat{b}$
    \item $\widehat{\bigvee \limits_{j \in J} a_j} = \bigvee \limits_{j \in J} \hat{a_j}$
  \end{enumerate}

  If $\mathcal{Q}$ is an unital quantale, then we express the neutral element as $\hat{\varepsilon} = \{ \langle b, c \rangle \: | \: b \cdot \varepsilon \leq c \} = \{ \langle b, c \rangle \: | \: b \leq c \}$.

  Also, one may prove that $\hat{a} \subseteq \hat{b}$ iff $a \leq b$.

  Let $\mathcal{S} \subseteq \mathcal{Q}$, so $I_{\mathcal{S}} a := \bigvee \{ s \: | \: s \in S, s \leq a \}$ is quantic conucleus. So, $\hat{\mathcal{S}} \subseteq \hat{\mathcal{Q}}$ is a subquantale of $\hat{\mathcal{Q}}$ by~\cref{prop1}.

  Let us define $\hat{I} \hat{a} := \bigvee \{ \hat{s} \: | \: \hat{s}  \in \hat{S}, \hat{s} \subseteq \hat{a} \}$, then

  $\begin{array}{lll}
  &\widehat{(I_{\mathcal{S}} a)} = \{ \langle p, q \rangle \: | \: p \leq I_{\mathcal{S}} a \cdot q \} = &\\
  &\:\:\:\: \text{By the condition}& \\
  &\{ \langle p, q \rangle \: | \: p \leq \bigvee \{ s \: | \: s \in \mathcal{S}, s \leq a \} \cdot q \} = &\\
  &\:\:\:\: \text{Homomorphism}& \\
  &\widehat{(\bigvee_{s \in S, s \leq a} s)} = & \\
  &\:\:\:\: \text{Homomorphism preserves sups}& \\
  &\bigvee_{s \in S, s \leq a} \hat{s} = & \\
  &\:\:\:\: \text{Unfolding}& \\
  &\bigvee \{ \hat{s} \: | \: s \in S, s \leq a \} = & \\
  &\:\:\:\: \text{Unfolding}& \\
  &\bigvee \{ \hat{s} \: | \: \hat{s} \in \hat{\mathcal{S}}, \hat{s} \subseteq \hat{a} \} = \widehat{I_{\mathcal{S}}} \hat{a}& \\
  \end{array}$

Thus, $\widehat{I_{s}} \hat{a} = \widehat{I_{\mathcal{S}} a}$.

\end{proof}

The representation theorem for quantales for quantic conuclei follows immediately from the lemma above.

\begin{theorem}
  Let $\mathcal{Q} = \langle \mathcal{Q}, I \rangle$ be a quantale with
  quantic conucleus. Then $\mathcal{Q}$ is isomorphic to relational quantale on $Q$ with some conucleus operator.
\end{theorem}

\begin{proof}
  Let $I$ be a quantic conucleus. By Lemma, $I a = \bigvee \{ q \in \mathcal{Q}_I \: | \: q \leq a \}$.
  One need to apply the previous statement with conucleus $I' a = \bigvee \{ q \in \mathcal{Q}_I \: | \: q \leq a \}$.
\end{proof}

This representation theorem might be extended categorically.
Let ${\bf Con Quant}$ be a category of (unital) quantales with quantic conuclei with homomorphisms between them.

${\bf R Con Quant}$ is a category of relational quantales with conuclei operators.
It is clear that any relational quantale is a (unital) quantale. Thus, ${\bf Con RQuant}$ is a subcategory of ${\bf Con Quant}$.

Let us show that ${\bf Con Quant}$ is equivalent to ${\bf Con RQuant}$.

\begin{theorem}
  There exists a functor $\mathfrak{F} : {\bf Con Quant} \to {\bf Con RQuant}$
\end{theorem}

\begin{proof}
  Let us define the following map:

  \begin{itemize}
    \item $\mathfrak{F} : \mathcal{Q} = \langle Q, \cdot, \bigvee, I \rangle \mapsto \hat{\mathcal{Q}} = \langle \mathcal{R}, \bigvee, \circ, \hat{I} \rangle$, \\ where $\mathcal{R} = \{ \hat{a} \: | \: a \in \mathcal{Q} \}$,
    $\mathcal{Q} \in Ob({\bf Con U Quant})$, and $\hat{a}$ is defined according to the construction provided in the proof of representation theorem
    \item Let $f \in Mor({\bf Con U Quant})$, that is, $f : \langle Q_1, \cdot, \bigvee, I_1 \rangle \to \langle Q_2, \cdot, \bigvee, I_2 \rangle$, \\ then $\mathfrak{F} : f \mapsto \hat{f}$, where $\hat{f}(\hat{a}) = \widehat{f(a)}$.
  \end{itemize}

  It is enough to check that $\mathfrak{F}(f)(\hat{I}(\hat{a})) = \hat{I}(\mathfrak{F}(f)(\hat{a}))$:

  $\begin{array}{lll}
  &\mathfrak{F}(f)(\hat{I}(\hat{a})) = \hat{f}(\hat{I}(\hat{a})) = \hat{f}(\widehat{I a}) = \widehat{f (I a)} =
  \hat{I} (\widehat{f (a))} = \hat{I}(\mathfrak{F}(f)(\hat{a}))&
  \end{array}$

\vspace{\baselineskip}

  The last equation follows from the fact that the category of quantales is equivalent to its subcategory of relational quantales \cite{BrownAndGurr1}.
\end{proof}

Finally, we formulate the theorem which claims that $SMALC_{\Sigma}$ is complete w.r.t relational quantales with quantic conuclei.

\begin{theorem}
  $SMALC_{\Sigma} \vdash \Gamma \rightarrow A$ iff $\Gamma \models_{\mathcal{Q}_R} A$ for some relational quantale $\mathcal{Q}_R$ with a family of quantic conuclei.
\end{theorem}

\begin{proof}
  Follows from soundness, completeness, and representation theorems.
\end{proof}

Thus, we proved that Lambek calculus is sound and complete with respect to quantales with quantic conuclei defined via subquantales. Also, we showed that any quantale with quantic conuclei is isomorphic to relational quantale with coclosure operators on its underlying set. The representation theorem allows us to claim that $SMALC_{\Sigma}$ is also relationally complete.

\section{Acknowledgements}

Author would like to thank Lev Beklemishev, Giuseppe Greco, Stepan Kuznetsov, Jean-Francois Mascari, Fedor Pakhomov, Danyar Shamkanov, Ilya Shapirovsky, and Andre Scedrov for advice, critique, and comments.

The research is supported by the Presidential Council, research grant MK-430.2019.1.


\addcontentsline{toc}{section}{References}

\begin{thebibliography}{}

  \bibitem{Ajd} Ajdukiewicz, K. \emph{Die syntaktische konnexit{\"a}t} Studia Philosophica, {\bf 1}, 1935, pp 1--27.

  \bibitem{Bar} Bar-Hillel, Y. \emph{A quasi arithmetical notation for syntactic description}, Language, {\bf 29}, 1953, pp 47--58.

  \bibitem{BrownAndGurr1} Brown C., Gurr D., \emph{A representation theorem for quantales}, Journal of Pure and Applied Algebra, {\bf 85} (1), 1993, pp. 27--42.

  \bibitem{BrownAndGurr2} Brown C., Gurr D., \emph{Relations and non-commutative linear logic}, Journal of Pure and Applied Algebra, {\bf 105} (2), 1995, pp. 117--136.

  \bibitem{Bus} Buszkowski W., \emph{Completeness Results for Lambek Syntactic Calculus}, Zeitschrift fur mathematische Logik und Grundlagen der Mathematik, {\bf 32}, 1986, pp. 13--28.

  \bibitem{Bus1} Buszkowski W., \emph{An Interpretation of Full Lambek Calculus in Its Variant without Empty Antecedents of Sequents},
In: Asher N., Soloviev S. (eds) Logical Aspects of Computational Linguistics. LACL 2014. Lecture Notes in Computer Science, {\bf
 8535}, 2014, pp. 30--43.

 \bibitem{Quantale18} Eklund, P., Guti\'{e}rrez Garcia, J., H\"{o}hle, U., Kortelainen, J., \emph{Semigroups in Complete Lattices.
Quantales, Modules and Related Topics}, Springer, Cham, 2018.

  \bibitem{Johnstone} Johnstone P., \emph{Stone Spaces}, Cambridge University Press, 1982

  \bibitem{KanovichEtc} Kanovich M., Kuznetsov S., Nigam V., Scedrov A., \emph{A Logical Framework with Commutative and Non-Commutative Subexponentials}, In: Galmiche D., Schulz S., Sebastiani R. (eds) Automated Reasoning. IJCAR 2018. Lecture Notes in Computer Science, {10900}, 2018, pp. 228--245.

  \bibitem{KanovichEtc1} Kanovich M., Kuznetsov S., Nigam V., Scedrov A., \emph{Subexponentials in non-commutative linear logic}, Mathematical Structures in Computer Science, 2018, 1--33.

  \bibitem{KanovichEtc2} Kanovich M., Kuznetsov S., Nigam V., Scedrov A., \emph{On the proof theory of non-commutative
  subexponentials}, Seminar in honor of the 60th birthday of Dale Miller (Paris, December 15--16, 2017), Universit\'e Paris Diderot, 2016, 10 p.

  \bibitem{Unit} Kuznetsov, S., \emph{Lambek Grammars with the Unit} In: de Groote, P., Nederhof, M.- J. (eds.) Formal Grammar 2010/2011. LNCS, {\bf 7395}, 2012, pp. 262--266.

  \bibitem{Lambek} Lambek, J., \emph{The mathematics of sentence structure}, American Mathematical Monthly,
{\bf 65}, 1958, pp. 154--170.

  \bibitem{LambekStar} Lambek, J., \emph{On the calculus of syntactic types} Structure of Language and
  Its Mathematical Aspects (Proc. Symposia Appl. Math., vol. 12), AMS, 1961, pp. 166--178.

  \bibitem{Morrill} Morrill G., Valentin O., \emph{Computation coverage of TLG: Nonlinearity}, In NLCS, 2015.

  \bibitem{NigamMiller} Nigam V., Miller D., \emph{Algorithmic specifications in linear logic with subexponentials}, In Proc. PPDP ’09, 2009, pp. 129--140.

  \bibitem{Paiva} de Paiva V., Eades III C., \emph{Dialectica Categories for the Lambek Calculus}, 2018, In Artemov S., Nerode A. (eds) Logical Foundations of Computer Science. LFCS 2018. Lecture Notes in Computer Science, {\bf 10703},
  Springer, 2018, pp 256--272.

  \bibitem{Pentus} Pentus M., \emph{Models for the Lambek calculus}, Annals of Pure and Applied Logic, {\bf 75}, No 1-2, 1995, pp. 179--213.

  \bibitem{Rasowa} Rasiowa, H., Sikorski R., \emph{The Mathematics of Metamathematics}, PWN-Polish Scientific Publishers, Warsaw, 1963.

  \bibitem{Myself} Rogozin, D., \emph{Quantale semantics for Lambek calculus with subexponentials}, TACL 2019. Abstracts, 2019, pp. 158--159.

  \bibitem{Rosenthal} Rosenthal K., \emph{Quantales and their Applications}, Pitman Research Notes in Mathematics, {\bf 234}, Longman
Scientific \& Technical, 1990, 165 pp.

  \bibitem{Wilde} Wilde O., \emph{The collected poems by Oscar Wilde}, Wordrworth Editions Ltd, 1994.

  \bibitem{Wurm} Wurm, C., \emph{Language-Theoretic and Finite Relation Models for the (Full) Lambek Calculus}, Journal of Logic, Language, and Information, {\bf 26}, 2017, pp. 179--214.

  \bibitem{Yetter} D. N. Yetter., \emph{Quantales and (noncommutative) linear logic}, Journal of Symbolic Logic, {\bf 55}(1), 1990, pp. 41--64.
\end{thebibliography}
\end{document}